\documentclass[10pt, a4paper]{article}
\usepackage{graphicx}
\long\def\omitit#1{}

\usepackage{amssymb,amsmath}
\usepackage{algorithm2e}
\usepackage{url}
\usepackage{setspace}
\usepackage{multirow}
\usepackage{amsthm}
\newtheorem{definition}{Definition}
\newtheorem{theorem}{Theorem}
\newtheorem{proposition}{Proposition}

\begin{document}

\title{Multi-sensor Information Processing using
Prediction Market-based Belief Aggregation}

\author{Janyl Jumadinova (1), Prithviraj Dasgupta (1) \\
  ((1) University of Nebraska at Omaha)
}

\maketitle

We consider the problem of information fusion from multiple sensors
of different types with the objective of improving the
confidence of inference tasks, such as object classification,
performed from the data collected by the sensors. We propose
a novel technique based on distributed belief
aggregation using a multi-agent prediction market
to solve this information fusion problem. To monitor
the improvement in the confidence of the
object classification as well as to dis-incentivize
agents from misreporting information, we have introduced a
market maker that rewards the agents
instantaneously as well as at the end of
the inference task, based on the quality
of the submitted reports.
We have implemented the market maker's reward calculation
in the form of a {\em scoring rule} and have shown analytically that it
incentivizes truthful revelation or accurate reporting
by each agent. We have experimentally verified
our technique for multi-sensor information fusion
for an automated landmine detection scenario.
Our experimental results show that, for
identical data distributions and settings, using our
information aggregation technique increases the accuracy of object
classification favorably
as compared to two other commonly used techniques for information fusion
for landmine detection.








\section{Introduction}
\label{sec_intro}
Information fusion from multiple sensors has
been a central research topic in sensor-based systems \cite{Waltz}
and recently several multi-agent techniques \cite{Vinyals11}
have been proposed to address this problem.
Most of the solutions for multi-sensor information fusion and
processing are based on Bayesian inference
techniques \cite{Makarenko06,Osborne08,Rosencrantz03}.
While such techniques have been shown to be very effective, we
investigate a complimentary problem where sensors
can behave in a self-interested manner. Such self-interested
behavior can be motivated by
malicious nodes that might have been planted into the system to
subvert its operation, or, by normal sensor nodes attempting to
give an illusion of efficient performance when they do
not have enough resources (e.g., battery power) to
perform accurate measurements. To address this problem,
we describe a market-based aggregation technique
called a prediction market for multi-sensor information fusion
that includes a utility driven mechanism to motivate each sensor,
through its associated agent, to reveal accurate reports.

To motivate our problem we describe a distributed
automated landmine detection scenario used for
humanitarian demining. An environment contains
different buried objects, some of which could
potentially be landmines. A set of robots,
each equipped with one of three types
of landmine detection sensor such as
a metal detector (MD), or a ground penetrating radar (GPR)
or an infra-red (IR) heat sensor, are deployed
into this environment.
Each robot is capable of perceiving certain
features of a buried object through its sensor such as
the object's metal content, area, burial depth, etc.
However, the sensors give noisy readings
for each perceived feature depending on the characteristics
of the object as well as on the characteristics
of the environment (e.g., moisture content, ambient
temperature, sunlight, etc.). Consequently,
a sensor that works well in one scenario, fails to detect
landmines in a different scenario, and,
instead of a single sensor, multiple sensors
of different types, possibly with different
detection accuracies can detect landmines with
higher certainty \cite{Gros}.
Within this scenario, the central question that we intend to
answer is: given an initial set of reports about the
features of a buried object, what is
a suitable set (number and type) of sensors to
deploy over a certain time window to the object,
so that, over this time window,
the fused information from the different sensors
successively reduces the uncertainty in determining
the object's type.

Our work in this paper is based on the insight that
the scenario illustrated above, of fusing information from
multiple sources to predict the outcome of an initially unknown object,
is analogous to the problem of aggregating the beliefs
of different humans to forecast the outcome
of an initially unknown event. Such forecasting
is frequently encountered in many problems such
as predicting the outcome of geo-political events,
predicting the outcome of financial instruments like stocks, etc.
Recently, a market-based model called {\em prediction
market} has been shown to be very successful in
aiding humans with such predictions
and with decision-making
\cite{Chen09,Chen,Othman,Wolfers04}.
Building on these models, in this paper, we describe a
multi-agent prediction market
for multi-sensor information fusion.
Besides being an efficient aggregation mechanism,
using prediction markets gives us several useful features -
a mathematical formulation called a {\em scoring rule}
that deters malicious sensors from misreporting
information, a regression-based belief update mechanism for
the sensor agents for incorporating the aggregated beliefs
(or information estimates) of other sensors into
their own calculation, and the ability to incorporate an autonomous
decision maker that uses expert-level domain knowledge
to make utility maximizing decisions
to deploy additional sensors appropriately to improve the detection
of an object.
Our experimental results illustrated
with a landmine detection scenario while using
identical data distributions and settings, show
that the information fusion performed using our technique
reduces the root mean squared error by $5-13\%$
as compared to a previously studied technique  for
landmine data fusion using the Dempster-Shafer theory \cite{Nada03} and
by $3-8\%$ using distributed data fusion technique \cite{Manyika95}.

\section{Related Work}
{\bf Multi-agent Information Fusion.}
Multi-agent systems have been used to solve various sensor network
related problems and an excellent overview is given in \cite{Vinyals11}.
In the direction of multi-sensor information processing, significant works
include the use of particle filters\cite{Rosencrantz03},
distributed data fusion (DDF) architecture along with
its extension, the Bayesian DDF \cite{Makarenko06,Manyika95},
Gaussian processes \cite{Osborne08} and mobile agent-based information
fusion \cite{Wu04}. For most of the application domains
described in these works such as water-tide height measurement,
wind speed measurement, robot tracking and localization, etc.,
self-interested behavior by the sensors is not considered a crucial problem.
For our illustrated application domain of landmine detection,
decision-level fusion techniques have been reported to be
amenable for scenarios where the sensor types are different from each other,
and, non-statistical decision-level fusion techniques, such as Dempster-Shafer theory \cite{Nada03},
fuzzy logic\cite{Cremer}, and rule-based fusion techniques \cite{Gros} have been reported to
generalize well. However, in contrast to our work,
these techniques assume that sensors are fully cooperative
and never behave self-interestedly by misreporting information.
In \cite{Mukherjee10}, the authors have observed
that most sensor-based information aggregation techniques either
do not consider malicious behavior or use high-overhead,
cryptographic techniques to combat it. To deter false
reports by sensor nodes in a data aggregation setting,
they propose various lower overhead reputation-based schemes.
Our prediction market-based information aggregation technique
is complimentary to such reputation-based aggregation techniques.

{\bf Decision-Making using Prediction Markets.}
A prediction market is a market-based aggregation mechanism that is
used to combine the opinions on the outcome of a future, real-world
event from different people, called the market's {\em traders}
and forecast the event's possible
outcome based on their aggregated opinion.
Recently, multi-agent systems have been used
\cite{Chen09,jumadinova10,Othman} to analyze the
operation of prediction markets, where the behaviors of the market's
participants are implemented as automated software agents.
The seminal work on prediction market analysis
\cite{Wolfers04} has shown that the mean belief
values of individual traders about the outcome of a future event
corresponds to the event's market price. The basic operation
rules of a prediction market are similar to those of a continuous
double auction, with the role of the auctioneer being
taken up by an entity called the {\em market maker}
that runs the prediction market. Hanson \cite{Hanson07} developed a
mechanism, called a
\emph{scoring rule}, that can be used by market makers
to reward traders for making and improving a
prediction about the outcome of an event, and,
showed that if a scoring rule is {\em proper}
or incentive compatible, then it can serve as
an automated market maker. Recently, authors
in \cite{Chen,Othman} have theoretically analyzed
the properties of prediction markets used for
decision making. In \cite{Othman}, the authors
analyzed the problem of a decision
maker manipulating a prediction market and
proposed a family of scoring rules to address
the problem. In \cite{Chen}, the authors
extended this work by allowing randomized decision rules,
considering multiple possible outcomes
and providing a simple test to determine whether the scoring
rule is proper for an arbitrary decision rule-scoring rule pair.
In this paper, we use a prediction market
for decision making, but in contrast to
previous works we consider that the
decision maker can make multiple, possibly
improved decisions over an event's duration,
and, the outcome of an event is decided
independently, outside
the market, and not influenced by the decision
maker's decisions. Another
contribution our paper makes is a new, proper
scoring rule, called the \emph{payment function}, that
incentivizes agents to submit truthful reports.

\section{Problem Formulation}
\label{sec_problem}
Let $L$ be a set of objects. Each object has certain
features that determine its type. We assume that
there are $f$ different features and $m$ different
object types. Let $\Phi=\{\phi_1, \phi_2,...,\phi_f\}$
denote the set of object features and
$\Theta = \{\theta_1, \theta_2, ..., \theta_m\}$ denote the set of
object types. The features of an object $l \in L$
is denoted by $l_\Phi \subseteq \Phi$ and its
type is denoted by $l_\theta \in \Theta$.
As illustrated in the example given in Section \ref{sec_intro},
$l_\Phi$ can be perceived, albeit with measurement
errors, through sensors, and,
our objective is to determine $l_\theta$ as accurately
as possible from
the perceived but noisy values of $l_\Phi$. Let
$\Delta(\Theta)  = \{(\delta(\theta_1), \delta(\theta_2), ...,
\delta(\theta_m)): \delta(\theta_i)\in [0,1], \sum_{i=1}^m
\delta(\theta_i) = 1\}$, denote
the set of probability distributions over the different
object types. For convenience of analysis,
we assume that when the actual type of object $l$,
$l_\theta = \theta_j$, its (scalar) type
is expanded into a $m$-dimensional probability vector
using the function
$vec: \Theta \rightarrow [0,1]^m: vec_j = 1, vec_{i \neq j}=0$,
which has $1$ as its $j$-th component corresponding
to $l$'s type $\theta_j$ and $0$ for all other
components.

Let $A$ denote a set of agents (sensors) and
$A^{t,l}_{rep} \subseteq A$ denote the subset of agents
that are able to perceive the object $l$'s
features on their sensors at time $t$.
Based on the perceived
object features, agent $a \in A_{rep}^{t,l}$ at time $t$
reports a belief as a
probability distribution over the set of object types,
which is denoted as $\mathbf{b}^{a,t,l} \in \Delta(\Theta)$.
The beliefs of all the agents are combined into a
composite belief, $\mathbf{B}^{t,l} = Agg_{a \in A_{rep}^{t,l}} (\mathbf{b}^{a,t,l})$,
and let $\hat{\Theta}^{t,l}: \mathbf{B}^{t,l} \rightarrow \Delta(\Theta)$
denote a function that computes a probability
distribution over object types based on
the aggregated agent beliefs.
Within this setting we formulate the object
classification problem as a decision making
problem in the following manner:
given an object $l$ and an initial aggregated
belief $\mathbf{B}^{t,l}$ calculated from one or more
agent reports for that object, determine
a set of additional agents (sensors) that need
to be deployed at object $l$ such that
the following constraint is satisfied:
\begin{equation}
\min \, RMSE\left(\hat{\Theta}^{t,l}, vec(l_\theta)\right), \quad \mbox{for} \quad
t =1, 2,.... T
\label{eqn_objectivefunction}
\end{equation}
where $T$ is the time window for classifying an object $l$
and RMSE is the root mean square error given by
$RMSE (\mathbf{x}, \mathbf{y}) = \frac{\mid \mid x - y \mid \mid}{\sqrt{m}}$.
In other words, at every time step $t$,
the decision maker tries to select
a subset of agents such that
the root mean square error (RMSE) between the
estimated type of object $l$ and its actual type
is successively minimized.

The major components of the object classification problem described
above consists of two parts: integrating the reports from the different
sensors and making sensor deployment decisions based on those reports so that
the objective function given in Equation \ref{eqn_objectivefunction} is satisfied.
To address the first part, we have used distributed
information aggregation with a multi-agent prediction market,
while for the latter we have used an expected utility maximizing
decision-making framework. A schematic showing the different
components of our system and their interactions is shown
in Figure \ref{fig_diagram} and explained in the following sections.

\begin{figure}[h!]
\begin{center}
\hspace{-0.2in}
    \includegraphics[width=3.2in,angle=-90]{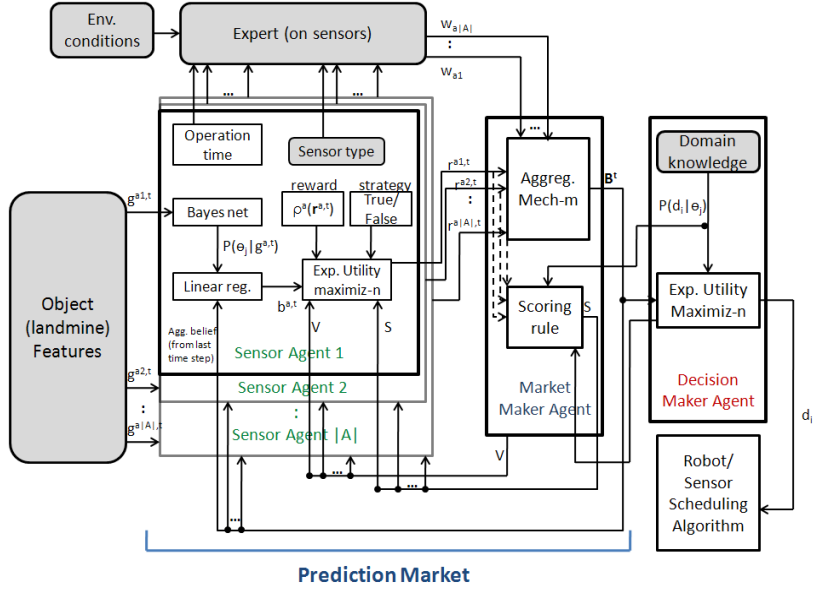}
\hspace{-0.2in}
\caption{The different components of the prediction
market for decision making and the interactions between them.}
\label{fig_diagram}
\end{center}
\end{figure}

\subsection{Sensor Agents}
\label{sec_agents}
As mentioned in Section \ref{sec_intro},
there is a set of robots in the
scenario and each robot has an on-board sensor
for analyzing the objects in the scenario.
Different robots can have different types of sensors
and sensors of the same type can have different
degrees of accuracy determined by their cost.
Every sensor is associated with a software agent
that runs on-board the robot and performs calculations
related to the data sensed by the robot's sensor.
In the rest of the paper,
we have used the terms sensor and agent interchangeably.
For the ease of notation, we drop the subscript $l$ corresponding
to an object for the rest of this section.
When an object is within the sensing range of
a sensor (agent) $a$ at time $t$, the sensor observes
the object's features and its agent receives this
observation in the form of an
information signal $g^{a,t} = <g_1,...,g_f>$
that is drawn from the space of information signals $G \subseteq \Delta(\Theta)$.
The conditional probability
distribution of object type $\theta_j$ given an information
signal $g \in G$,
$P(\theta_j|g): G \rightarrow [0,1]$,
is constructed using domain knowledge \cite{Cremer,Nada03,Nada01}
within a Bayesian network and is made available to each agent.
Agent $a$ then updates its belief distribution $\mathbf{b}^{a,t}$
using the following equation:
\begin{equation}
\mathbf{b}^{a,t} = w_{bel} \cdot \mathbf{P}(\Theta|g^{a,t}) + (1-w_{bel}) \cdot \mathbf{B}^t,
\label{belief_update}
\end{equation}
where $\mathbf{B}^t$ is the belief value vector aggregated from
all sensor reports.

{\bf Agent Rewards.}
Agents behave in a self-interested manner to ensure that
they give their `best' report using their available resources
including sensor, battery power, etc. However, some agents can behave maliciously,
either being planted or compromised to infiltrate the system and subvert
the object classification process, or, they might be trying to
give an illusion of being efficient when they do not have sufficient resources
to give an accurate report. An agent $a$ that submits a report at time $t$,
uses its belief distribution $\mathbf{b}^{a,t}$ to calculate the report
$\mathbf{r}^{a,t} = <r_1^{a,t},...,r_m^{a,t}> \in \Delta(\Theta)$.
An agent can have two strategies to make this report - truthful
or malicious. If the agent is truthful, its report corresponds to
its belief, i.e., $\mathbf{r}^{a,t} = \mathbf{b}^{a,t}$. But
if it is malicious,
it manipulates its report to reveal an inaccurate belief.
Each agent $a$ can update its report $\mathbf{r}^{a,t}$ within
the time window $T$ by obtaining new measurements from the object
and using Equation \ref{belief_update} to update its belief.
The report from an agent $a$ at time $t$ is analyzed
by a human or agent expert \cite{Nada03} to
assign a weight  $w^{a,t}$ depending
on the current environment conditions and
agent $a$'s sensor type's accuracy under those environment
conditions (e.g., rainy weather reduces the weight
assigned to the measurement from an IR heat sensor,
or, soil that is high in metal content reduces
the weight assigned to the measurement from an metal detector).

To motivate an agent to submit reports,
an agent $a$ gets an instantaneous reward, $\rho^{a,t}$,  from the market maker for the report
$\mathbf{r}^{a,t}$ it submits  at time $t$, corresponding to
its instantaneous utility, which is given by the following equation:
\begin{equation}
\rho^{a,t} = V(n^{t'=1..t}) - C^a(\mathbf{r}^{a,t}),
\label{eqn_reward}
\end{equation}
where $V(n^{t'=1..t})$ is the value for making a report
with $n^{t'=1..t}$ being the number of times the agent $a$ submitted a report
up to time $t$, and, $C^a(\mathbf{r}^{a,t})$
is the cost of making report $\mathbf{r}^{a,t}$ for agent $a$
based on the robot's expended time, battery power, etc.
We denote the agent's value for each report $V(n^{t'=1..t})$ as a
constant-valued function up to a certain threshold and a linearly
decreasing function thereafter, to de-incentivize agents from
making a large number of reports. Agent $a$'s value function is
given by the following equation:
\begin{displaymath}
   V(n^{t'=1..t}) = \left\{
     \begin{array}{ll}
       \nu & , n^{t'=1..t} \leq n^{threshold}\\
       \frac{\nu(n^{t'=1..t} - n^{max})}{(n^{threshold} - n^{max})} & , otherwise
     \end{array}
   \right.
\end{displaymath}
where $\nu \in \mathbb{Z^+}$, is a constant value that $a$ gets by submitting
reports up to a threshold,
$n^{threshold}$ is the threshold corresponding to the number of reports
$a$ can submit before its report's value starts decreasing, and, $n^{max}$ is the
maximum number of reports agent $a$ can submit before $V$ becomes negative.
Finally, to determine its strategy while submitting its report, an agent selects the
strategy that maximizes its expected utility obtained from its cumulative
reward given by Equation \ref{eqn_reward} plus an expected value
of its final reward payment if it continues making similar reports
up to the object's time window $T$.

\subsection{Decision Maker Agent}
\label{sec_dm}
The decision maker agent's task is to use the composite
belief about an object's type, $\mathbf{B}^t$, given by the prediction
market, and take actions to deploy additional robots(sensors)
 based on the value of the objective function
given in Equation \ref{eqn_objectivefunction}.
Let $AC$ denote a set of possible actions corresponding to deploying
a certain number of robots, and $D=\{d_1,...d_h\}: d_i \in Ac \subseteq AC$
denote the decision set of the decision maker. The decision function
of the decision maker is given by $dec: \Delta(\Theta) \rightarrow D$.
Let $\mathbf{u}^{dec}_j \in \mathbb{R}^m$ be the utility that the decision maker receives
by determining an object to be of type $\theta_j$ and
let $P(d_i|\theta_j)$ be the probability that the decision maker makes decision
$d_i \in D$ given object type $\theta_j$. $P(d_i|\theta_j)$ and $\mathbf{u}^{dec}_j$ are constructed
using domain knowledge \cite{Cremer,Nada01,Nada03}. Given the
aggregated belief distribution $\mathbf{B}^t$ at time $t$,
the expected utility to the decision maker for taking decision $d_i$ at time $t$ is
then $EU^{dec}(d_i,\mathbf{B}^t) = \sum_{j=1}^m  P(d_i|\theta_j) \cdot \mathbf{u}^{dec}_j \cdot \mathbf{B}^t$.
The decision that the decision maker takes at time $t$, also
called its {\em decision rule}, is the one
that maximizes its expected utility and is given by:
$d^t = \arg \max_{d_i} EU^{dec}(d_i,\mathbf{B}^t)$.

\subsection{Prediction Market}
A conventional prediction market uses the aggregated beliefs
of the market's participants or traders about the outcome of a future event,
to predict the event's outcome. The outcome of an event is
represented as a binary variable (event happens/does not happen).
The traders observe information related to the event
and report their beliefs, as probabilities
about the event's outcome. The market maker
aggregates the traders' beliefs and uses
a scoring rule to determine a payment or payoff
that will be received by each reporting trader.
In our multi-agent prediction market,
traders correspond to sensor agents, the market maker
agent automates the calculations on behalf of the conventional
market maker, and, an event in the conventional market corresponds to
identifying the type of a detected object. The time window $T$
over which an object is sensed is called the duration
of the object in the market. This time window is divided into discrete
time steps, $t=1, 2...T$. During each time step, each sensor agent observing
the object submits a report about the object's type to the market
maker agent. The market maker agent performs two functions
with these reports. First, at each time step $t$,
it aggregates the agent reports into an aggregated belief
about the object, $\mathbf{B}^t \in \Delta(\Theta)$. Secondly,
it calculates and distributes payments for the sensor agents.
It pays an immediate but nominal
reward to each agent for its report at time step $t$ using
Equation \ref{eqn_reward}. Finally, at the
end of the object's time window $T$, the market maker
also gives a larger \emph{payoff} to each agent that contributed
towards classifying the object's type.
The calculations and analysis related to these two functions
of the market maker agent are described in the following sections.

\textbf{Final Payoff Calculation.} The payoff calculation for
a sensor agent is performed by the market maker using a {\em decision scoring rule}
at the end of the object's time window.
A decision scoring rule \cite{Chen} is defined as any real
valued function that takes the agents' reported beliefs,
the realized outcome and the decisions
made by the decision maker as input, and
produces a payoff for the agent for its reported beliefs,
i.e. $S: \Delta(\Theta) \times \Theta \times D \longrightarrow \mathbb{R}$.
We design a scoring rule for decision making that is  based on how much agent $a$'s final report helped the decision
maker to make the right decisions throughout the duration of the prediction market
and by how close the agent $a$'s final report is to actual object type.
Our proposed scoring rule for decision making given that object's true
type is $\theta_j$ is given in Equation \ref{scoring_rule}:
\begin{equation}
S(r_j^{a,t},d^{[1:t]},\theta_j) = \varpi(d^{[1:t]}, \theta_j) log\left(r_j^{a,t}\right),
\label{scoring_rule}
\end{equation}
where,  $r_j^{a,t}$ is the reported belief that agent $a$ submitted at time $t$ for object type $\theta_j$,
$d^{[1:t]}$ is the set consisting of all the decisions that the decision
maker took related to the object up to the current time $t$,
$\theta_j$ is the object's true type that was revealed at the end the
object's time window, $log\left(r_j^{a,t}\right)$ measures the
goodness of the report at time $t$ relative to the true object type $\theta_j$,
and, $\varpi(d^{[1:t]},\theta_j)$ is the weight, representing how good all the decisions
the decision maker took up to time $t$ were compared to the true object type $\theta_j$.
$\varpi(d^{[1:t]},\theta_j)$ is determined by the decision maker and made
available to the agents through the market maker. We
assume that $\varpi(d^{[1:t]},\theta_j) = \sum_{i=1}^{t} P(d_i \mid \theta_j) \cdot \mathbf{u}_j^{dec}$, which gives
the expected utility of the decision maker agent for making decision $i$ when
the true type of the object is $\theta_j$.

\textbf{Aggregation.}
Since a sensor agent gets paid both through its immediate rewards
for making reports during the object's time window
and through the scoring rule function for
decision making at the end of the object's time window, we define the
total payment that the agent has received by the end
of the object's time window as a \emph{payment function}.

\begin{definition}
A function $\Psi(\mathbf{r}^{a,t}, d^{[1:t]}, \theta_j, n^{t'=1..t})$ is called a \textbf{payment
function} if each agent $a$'s total received payment at the end of the object's time window
(when $t=T$) is
\begin{equation}
\Psi(\mathbf{r}^{a,t}, d^{[1:t]}, \theta_j, n^{t'=1..t}) = \sum_{k=1}^t \rho^{a,k} + S(r_j^{a,t},d^{[1:t]},\theta_j) \\
\label{payment}
\end{equation}
where $\rho^{a,k}$, $S(r_j^{a,t},d^{[1:t]},\theta_j)$ and their
components are defined as in Equations \ref{eqn_reward} and \ref{scoring_rule}.
\end{definition}

Let $\Psi^{ave}$ denote a weighted average of the payment function
in Equation \ref{payment} over all the reporting agents,
using the report-weights assigned by the expert in Section \ref{sec_agents},
as given below:
\begin{align}
 \Psi^{ave}(\mathbf{r}^{A_{rep}^t,t}, d^{[1:t]}, \theta_j, n^{A_{rep}^t,t}) &= \sum_{k=1}^t \sum_{a \in A^t_{rep}}w^{a,k} \rho^{a,k}\\
 &+ \varpi(d^{[1:t]}, \theta_j)\sum_{a \in A^t_{rep}}  w^{a,t} log\left(r_j^{a,t}\right)\nonumber,
\label{eqn_ave_payment}
\end{align}
where $A_{rep}^t$ is the subset of agents that are able to perceive object feature at time $t$ and
$w^{a,k}$ is the weight assigned to agent $a$ at time $k$ by the expert.
To calculate an aggregated belief value in a prediction market,
Hanson \cite{Hanson07} used the generalized inverse function of the scoring rule.
Likewise, we calculate the aggregated
belief for our market maker agent
by taking the generalized inverse of the average payment function
given in Equation \ref{eqn_ave_payment}:
\begin{align}
B^t_j  &= Agg_{a \in A_{rep}^t} (\mathbf{b}^{a,t})\\
&= \frac{\frac{exp\left(\Psi^{ave}- \sum_{k=1}^t \sum_{a \in A^t_{rep}}w^{a,k}\rho^{a,k}\right)}{\varpi(d^{[1:t]},\theta_j)}}{\frac{\sum_{\theta_j=\theta_1}^{\theta_m} exp\left(\Psi^{ave} - \sum_{k=1}^t \sum_{a \in A^t_{rep}}w^{a,k}\rho^{a,k}\right)}{\varpi(d^{[1:t]},\theta_j)}}
\label{market_distr}
\end{align}
where $B^t_j \in \mathbf{B}^t$ is the $j$-th  component of the
aggregated belief for object type $\theta_j$.
The aggregated belief vector, $\mathbf{B}^t$, calculated by the market maker agent
is sent to the decision maker agent so that it can calculate
its expected utility given in Section \ref{sec_dm},
as well as, sent back to each sensor agent that reported the object's type
till time step $t$, so that the agent can refine its future reports,
if any, using this aggregate of the reports from other agents.

\section{Payment function: Properties \\and Characteristics}
In this section we first show that the payment function is proper, or incentive compatible.
Then we show that when the market maker uses this payment function to reward each agent for
its reported beliefs,
 reporting beliefs truthfully is the optimal strategy for each agent.

We can characterize a proper payment function similar to a proper scoring rule.
\begin{definition}
A payment function $\Psi$ is proper, or incentive compatible, if
\begin{equation}
\Psi(\mathbf{b^{a,t}}, d^{[1:t]}, \theta_j, n^{t'=1..t}) \geq \Psi(\mathbf{r^{a,t}}, d^{[1:t]}, \theta_j, n^{t'=1..t}),
\label{proper}
\end{equation}
$\forall \mathbf{b^{a,t}}, \mathbf{r^{a,t}} \in \Delta(O).$
\end{definition}
$\Psi$ is strictly proper if Equation \ref{proper} holds with equality, i.e.,
iff $\mathbf{b^{a,t}} = \mathbf{r^{a,t}}$.

Payment functions can be shown to be proper by representing
them using convex functions \cite{Chen,Gneiting}. To show
that our payment function in Equation \ref{payment} is
proper, we characterize it in terms of a convex function,
as shown below:

\begin{theorem}
A payment function $\Psi$ is proper for decision making if
\begin{equation}
\Psi(\mathbf{r^{a,t}}, d^{[1:t]}, \theta_j, n^{t'=1..t}) = G(\mathbf{r^{a,t}}) - G'(\mathbf{r^{a,t}}) \cdot (\mathbf{r^{a,t}}) +\frac{G'_{i,j}(\mathbf{r^{a,t}})}{P(d_i|\theta_j)},
\label{convex_proper}
\end{equation}
 where $G(\mathbf{r^{a,t}})$ is a convex function and $G'(\mathbf{r^{a,t}})$ is
 a subgradient of $G$ at point $\mathbf{r^{a,t}}$ and $P(d_i|\theta_j) > 0$.
\label{G_form}
\end{theorem}
\begin{proof}
Consider a payment function $\Psi$ satisfying Equation \ref{convex_proper}.
We will show that $\Psi$ must be proper for decision making. We will drop the agent
and time subscripts in this proof, and also we will write $\Psi(\mathbf{r},d^{[1:t]})$ (or its element $\Psi(r_j, d_i|\theta_j)$)
instead of full $\Psi(\mathbf{r^{a,t}}, d^{[1:t]}, \theta_j, n^{t'=1..t})$.
\begin{align*}
EU(\mathbf{b}, \mathbf{b}) &= \sum_{i=1}^h \sum_{j=1}^m P(d_i|\theta_j) b_j \Psi(b_j,d_i|\theta_j) \\
&= \sum_{i=1}^h \sum_{j=1}^m P(d_i) b_j \left( G(\mathbf{b}) - G'(\mathbf{b}) \cdot \mathbf{b} + \frac{G'_{i,j}(\mathbf{b})}{P(d_i|\theta_j)}\right)\\
&= G(\mathbf{b}) - G'(\mathbf{b})\cdot \mathbf{b} + \sum_{i=1}^h \sum_{j=1}^m G'_{i,j}(\mathbf{b}) b_j\\
&= G(\mathbf{b}) - G'(\mathbf{b})\cdot \mathbf{b} + G'(\mathbf{b})\cdot \mathbf{b} = G(\mathbf{b}).
\end{align*}
 Since $G$ is convex and $G'$ is its subgradient, we have
\begin{align*}
EU(\mathbf{b}, \mathbf{r}) &= \sum_{i=1}^h \sum_{j=1}^m P(d_i|\theta_j) b_j \Psi(r_j,d_i|\theta_j) \\
&= \sum_{i=1}^h \sum_{j=1}^m P(d_i|\theta_j) b_j \left( G(\mathbf{r}) - G'(\mathbf{r}) \cdot \mathbf{r} + \frac{G'_{i,j}(\mathbf{r})}{P(d_i|\theta_j)}\right)\\
&= G(\mathbf{r}) - G'(\mathbf{r}) (\mathbf{b} - \mathbf{r}) \\
&\leq G(\mathbf{b}) = EU(\mathbf{b}, \mathbf{b}).
\end{align*}
Thus, $\Psi$ is a proper payment function for decision making. $\Psi$ is strictly
proper payment function and the inequality is strict if $G$ is a strictly convex function.
\end{proof}

\begin{proposition}
The payment function given in Equation \ref{payment} is proper.
\end{proposition}
\begin{proof}
Let $G(\mathbf{b}) = EU(\mathbf{b}, \mathbf{b})$ and  \\
$G'_{i,j}(\mathbf{b}) = P(d_i|\theta_j) \Psi(\mathbf{b},d^{[1:t]},\theta_j, n^{t'=1..t})$.
Then we can write the payment function as
\begin{align*}
\Psi(\mathbf{b},d^{[1:t]},\theta_j,n^{t'=1..t}) &= \sum_{i=1}^h \sum_{j=1}^m P(d_i|\theta_j) b_j \Psi(\mathbf{b},d^{[1:t]},\theta_j,n^{t'=1..t}) \\
&- \mathbf{b} \cdot \sum_{i=1}^h \sum_{j=1}^m P(d_i|\theta_j) \Psi(\mathbf{b},d^{[1:t]},\theta_j,n^{t'=1..t}) \\
&+\frac{\Psi(\mathbf{b},d^{[1:t]},\theta_j,n^{t'=1..t})  \cdot P(d_i|\theta_j)}{P(d_i|\theta_j)}    \\
&= \Psi(\mathbf{b},d^{[1:t]},\theta_j,n^{t'=1..t}).
\end{align*}
We can clearly see that the payment function can be written in the form given in Equation \ref{convex_proper}
from Theorem \ref{G_form}. Therefore, the payment function $\Psi$ given in Equation
\ref{payment} is a proper payment function.
\end{proof}

 {\bf Agent Reporting Strategy.} Assume that agent $a$'s report at time $t$ is its final report, then its utility function can be written as
$u_j^{a,t} = \sum_{k=1}^t \rho^{a,k} + S(r_j^{a,t},d^{[1:t]},\theta_j)$.
Then, agent $a$'s expected utility for object type $\theta_j$ given its reported belief for object type $\theta_j$, $r_j^{a,t}$,
and its true belief about object type $\theta_j$, $b_j^{a,t}$ at time $t$ is
\begin{align}
EU_j^{a}(r_j^{a,t},b_j^{a,t}) &= \sum_{i=1}^h P(d_i|\theta_j) b_j^{a,t} u_j^{a,t} \\
&= \sum_{i=1}^h P(d_i|\theta_j) b_j^{a,t} \left(\sum_{k=1}^t \rho^{a,k} + S(r_j^{a,t},d^{[1:t]},\theta_j)\right) \nonumber,
\label{eqn_agentutility}
\end{align}
where $P(d_i|\theta_j)$ is the probability that the decision maker takes decision
$d_i$ when the object's type is $\theta_j$.

\begin{proposition}
If agent $a$ is paid according to $\Psi$, then it reports its beliefs about the object types truthfully.
\end{proposition}
\begin{proof}
Sensor agent $a$ wants to maximize its expected utility function and solves the following program
\begin{equation*}
\arg \max_{\mathbf{r}} \left(\sum_{i=1}^h \sum_{j=1}^m P(d_i|\theta_j) b_j^{a,t} \left[\sum_{k=1}^t \rho^{a,k} + \varpi(d^{[1:t]},\theta_j) log\left(r_j^{a,t}\right)\right]\right),
\end{equation*}
s.t. $\sum_{j=1}^m r_j^{a,t} = 1$.\\
The Lagrangian is
\begin{align*}
L(\mathbf{r},\lambda) &= \left(\sum_{i=1}^h \sum_{j=1}^m P(d_i|\theta_j) b_j^{a,t} \left[\sum_{k=1}^t \rho^{a,k} + \varpi(d^{[1:t]},\theta_j) log\left(r_j^{a,t}\right)\right]\right) \\
&- \lambda \left(\sum_{j=1}^m r_j^{a,t} - 1\right).
\end{align*}
The first order conditions are
\begin{align*}
&\frac{\partial L}{\partial r_j} = \sum_{i=1}^h \sum_{j=1}^m P(d_i|\theta_j) b_j^{a,t} \frac{\varpi(d^{[1:t]},\theta_j)}{r_j^{a,*}} - \lambda = 0\\
&\Rightarrow r_j^{a,*} = \frac{\varpi(d^{[1:t]},\theta_j) b_j^{a,t} \sum_{i=1}^h P(d_i|\theta_j)}{\lambda}\\
&\frac{\partial L}{\partial \lambda} = -\sum_{j=1}^m r_j^{a,t} + 1 = 0.\\
\end{align*}

Substituting $r_j^{a,*}$ into the second equation above, we have
\begin{align*}
&\frac{\varpi(d^{[1:t]},\theta_j) b_j^{a,t} \sum_{i=1}^h P(d_i|\theta_j)}{\lambda} = 1\\
&\lambda = \varpi(d^{[1:t]},\theta_j) \sum_{i=1}^h P(d_i|\theta_j)\\
&r_j^{a,*} = b_j^{a,t}.
\end{align*}
\end{proof}

\section{Experimental Results}
\begin{table}[h!]
\begin{center}
\begin{tabular}{|l|l|}
\hline
\textbf{Name} & \textbf{Value}\\
\hline
Object types & mine, metallic object(non-mine), \\
& non-metallic object(non-mine)\\
\hline
Features & metallic content, object's area, \\
& object's depth, sensor's position\\
\hline
Sensor types & MD, IR, and GPR\\
\hline
Max no. of sensors & $10$ ($5$MD,$3$IR,$2$GPR)\\
\hline
Max no. of decisions & $14$\\
\hline
$T$ (object identification window)& $10$\\
\hline
$\nu$ (agent's value if $n^{t'=1..t}\leq n^{threshold}$) & $5$\\
\hline
$n_{max}$ (max no. of reports before value & $20$\\
 is negative)& \\
\hline
$n_{threshold}$ (no. of reports before agent's value  & $5$\\
is less than $\nu$)  & \\
\hline
\end{tabular}
\caption{Parameters used for our simulation experiments.}
\label{table_sim_parms}
\end{center}
\end{table}
We have conducted several experiments using our aggregation
technique for decision-making
within a multi-sensor landmine detection scenario described
in Section \ref{sec_intro}.
Our environment contains different buried objects, some of which
are landmines. The true types of the objects are randomly determined at
the beginning of the simulation.
Due to the scarcity of real data related to landmine detection,
we have used the domain knowledge
that was reported in \cite{Cremer, Nada03, Nada01} to determine
object types, object features,
sensor agents' reporting costs, decision maker agent's decision set,
decision maker agent's utility of
determining objects of different types,
and, to construct the probability distributions for
$P(\theta_j|g)$ and $P(d_i|\theta_j)$.
We report simulation results for root mean squared error (RMSE) defined in
Section\ref{sec_problem} and also for number of sensors over time,
cost over object types, and average utility of the sensors over time.

{\bf Compared Techniques.}
For comparing the performance of our prediction
market based object classification techniques,
we have used two other well-known techniques for
information fusion: (a) Dempster-Shafer (D-S) theory
for landmine classification \cite{Nada03}, where
a two-level approach based on belief functions is used.
At the first level, the detected object is
classified according to its metal content. At
the second level the chosen level of metal content
is further analyzed to classify
the object as a landmine or a friendly object.
The belief update of the sensors that we used for D-S method
is the same one we have described in Section \ref{sec_agents}.
(b) Distributed Data Fusion (DDF) \cite{Manyika95},
where sensor measurements are refined over
successive observations using a temporal, Bayesian
inference-based, information filter.
To compare DDF with our prediction market-based
technique, we replaced our belief aggregation mechanism
given in Equation \ref{market_distr} with
a DDF-based information filter.
We compare our techniques using some
standard evaluation metrics from multi-sensor
information fusion \cite{Osborne08}:
root mean squared error (RMSE)
defined as in Section \ref{sec_problem}, normed
mean squared errors (NMSE) calculated as:\\
$NMSE^t(\hat{\Theta}^t-vec(\theta_j)) = 10\, log_{10} \frac{\frac{1}{m}\sum_{j=1}^m \left(\hat{\Theta}_j^t,vec(\theta_j)\right)^2}{\frac{1}{m} \left(\sum_{j=1}^m vec(\theta_j)^2\right) - \left(\frac{1}{m} \sum_{j=1}^m vec(\theta_j)\right)^2}$, and,
the information gain, also known as Kullback-Leibler divergence
and relative entropy, calculated as:\\
$D^t_{KL}(\hat{\Theta}^t||vec(\theta_j) = \sum_{j=1}^m \hat{\Theta}_j^t log\left(\frac{\hat{\Theta}_j^t}{vec(\theta_j)}\right)$.\\
$\hat{\Theta}^t$ was calculated using D-S,
DDF, and our prediction market technique ($\hat{\Theta}^t = \mathbf{B}^t$).
\begin{table}[h]
\begin{center}
\begin{tabular}{|l|l|l|l|l|}
\hline
\textbf{Object type} & \textbf{Time}& \textbf{PM} & \textbf{DDF} & \textbf{D-S}\\
 & \textbf{steps}&  &  &  \\
\hline
\cline{1-5}
\textbf{Mine} & 1 & 1($1$MD) & 1($1$MD)& 1($1$MD)\\
\hline
 & 2 & 3($1$MD,$1$IR) & 3($1$MD,$1$GPR) & 3($1$IR,$1$GPR)\\
\hline
 & 3 & 4($1$GPR) & 5($1$MD,$1$IR) & 4($1$MD)\\
\hline
 & 4 & 5($1$MD) & 6($1$IR) & 5($1$MD)\\
\hline
 & 5 & 6($1$MD) & 7($1$$1$MD) & 6($1$IR)\\
\hline
 & 6 & 7($1$IR) & 8($1$MD) & 7($1$IR)\\
\hline
 & 7 & - & 9($1$IR) & 8($1$MD)\\
\hline
\cline{1-5}
\textbf{Metallic} or & 1 & 1($1$MD) & 1($1$MD)& 1($1$MD)\\
\hline
\textbf{Friendly} for D-S & 2 & 3($1$MD,$1$IR) & 4($1$MD,$1$IR,$1$GPR) & 3($1$IR,$1$GPR)\\
\hline
 & 3 & 4($1$GPR) & 5($1$MD) & 4($1$MD)\\
\hline
 & 4 & 5($1$MD) & 6($1$IR) & 5($1$IR)\\
\hline
 & 5 & 6($1$IR) & 7($1$MD) & 6($1$$1$IR)\\
\hline
 & 6 & 7($1$MD) & 8($1$IR) & 7($1$MD)\\
\hline
 & 7 & 8($1$IR) & 9($1$GPR) & 8($1$MD)\\
 \hline
 & 8 & - & 9($1$MD) & 8($1$MD)\\
\hline
\cline{1-4}
\textbf{Non-metallic} & 1 & 1($1$MD) & 1($1$MD) \\
\cline{1-4}
& 2 & 2($1$MD) & 2($1$IR)\\
\cline{1-4}
& 3 & 3($1$IR) & 3($1$MD) \\
\cline{1-4}
 & 4 & 4($1$MD) & 4($1$GPR) \\
\cline{1-4}
 & 5 & 5($1$IR) & 5($1$MD) \\
\cline{1-4}
 & 6 & 6($1$MD) & 6($1$IR) \\
\cline{1-4}
 & 7 & - & 7($1$MD) \\
\cline{1-4}
\end{tabular}
\caption{Different number of sensors and the sensor types deployed over time by a decision maker to classify different types of objects.}
\label{table_num_sensors}
\end{center}
\end{table}
Since the focus of our work is on the quality of information fusion,
we will concentrate on describing
the results for one object. We assume that there are three types of sensors,
MD (least operation cost, most noisy), IR (intermediate operation cost, moderately noisy),
and GPR (expensive operation cost, most accurate).
We also assume that there are a total of $5$ MD sensors, $3$ IR sensors, and $2$ GPR sensors
available to the decision maker for classifying this object.
Initially, the object is detected using one MD sensor. Once the object is
detected, the time window in the prediction market for identifying the object's type
starts. The MD sensor sends its report to the market
maker in the prediction market and the decision maker makes its first decision
based on this one report. We assume that decision maker's decision
(sent to the robot/sensor scheduling algorithm in Figure \ref{fig_diagram})
is how many ($0-3$) and what type (MD,IR,GPR) of sensors to send to the site
of the detected object subsequently. We have considered a set of $14$
out of all the possible decisions under this setting.
From \cite{Nada01}, we derive four object features, which are metallic content,
area of the object, depth of the object, and the position of the sensor.
Combinations of the values of these four features constitute the signal set $G$
and at each time step, a sensor perceiving the object receives a signal $g \in G$.
The value of the signal also varies based on the robot/sensor's current position
relative to the object.
We assume that the identification of an object stops
and the object type is revealed when either $\mathbf{B}^t_j \geq 0.95$, for any $j$,
or after $10$ time steps.
The default values for all domain related parameters are shown in Table \ref{table_sim_parms}. All of our results were
averaged over $10$ runs and the error bars indicate the standard deviation over the number of runs.
\begin{figure*}[ht!]
\begin{center}
\begin{tabular}{llll}
\hspace{-0.15in}
    \includegraphics[width=1.5in]{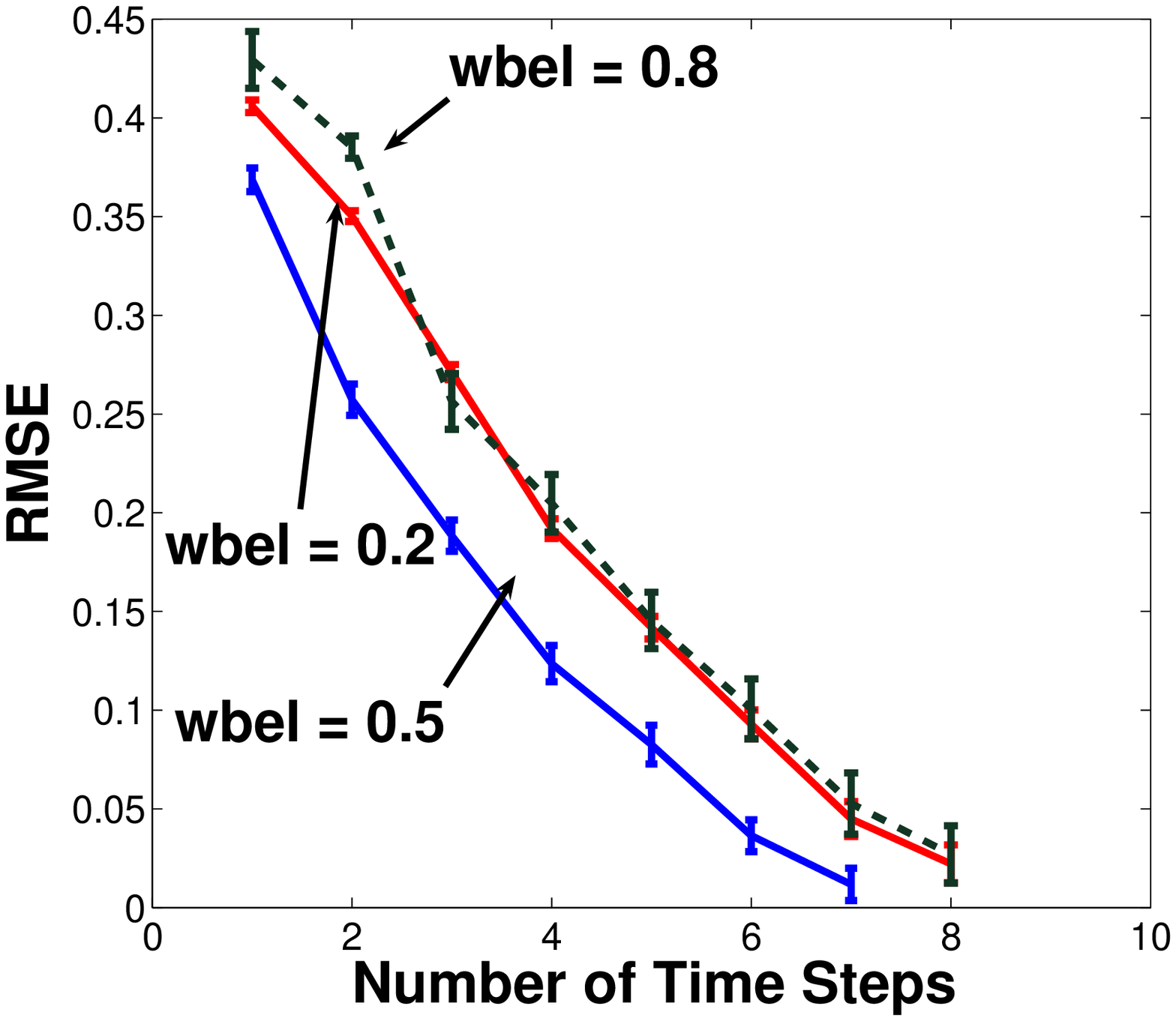}
    \hspace{-0.19in}
    &
    \hspace{-0.19in}
    \includegraphics[width=1.5in]{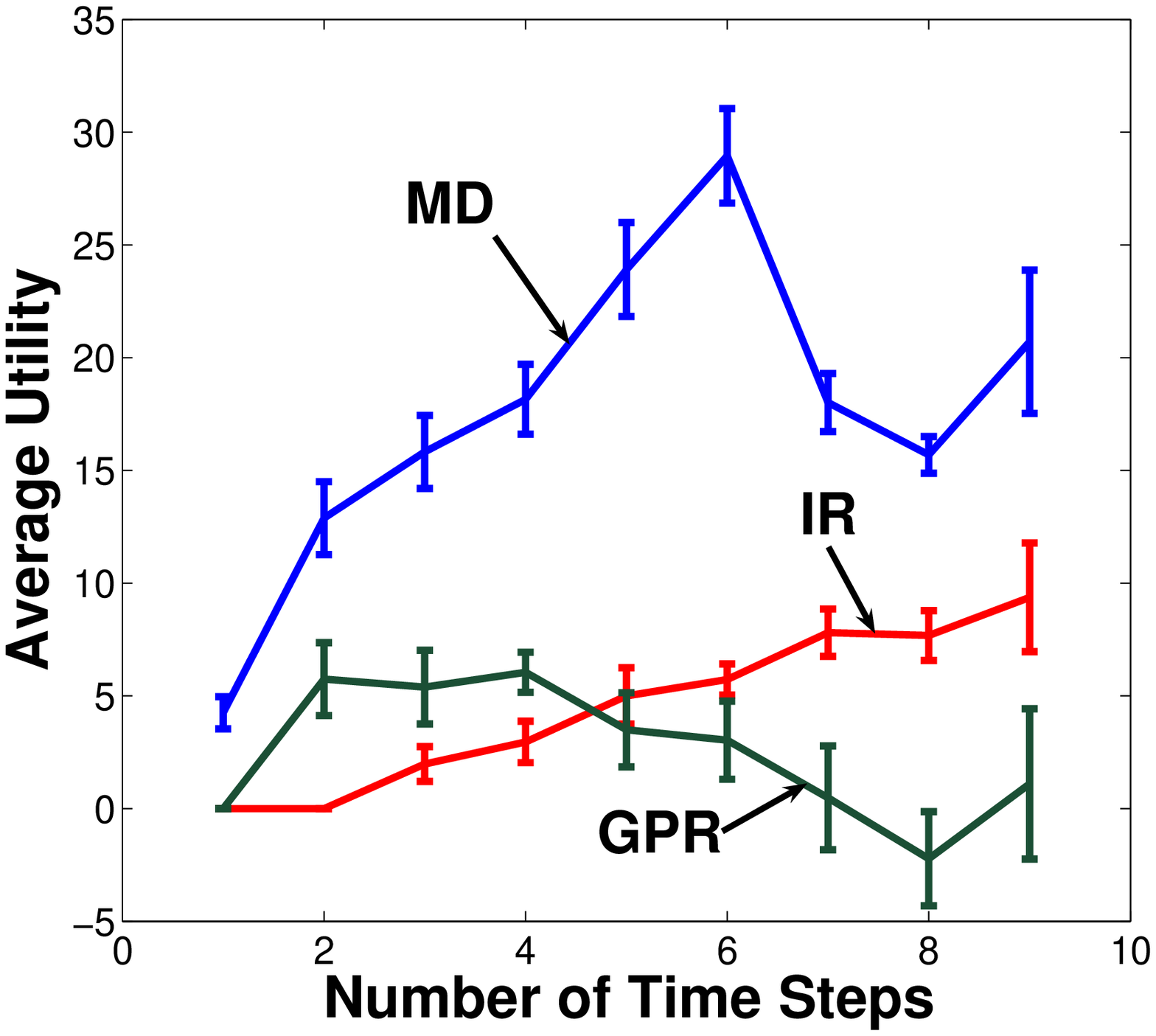}
    \hspace{-0.19in}
    &
    \hspace{-0.19in}
    \includegraphics[width=1.5in]{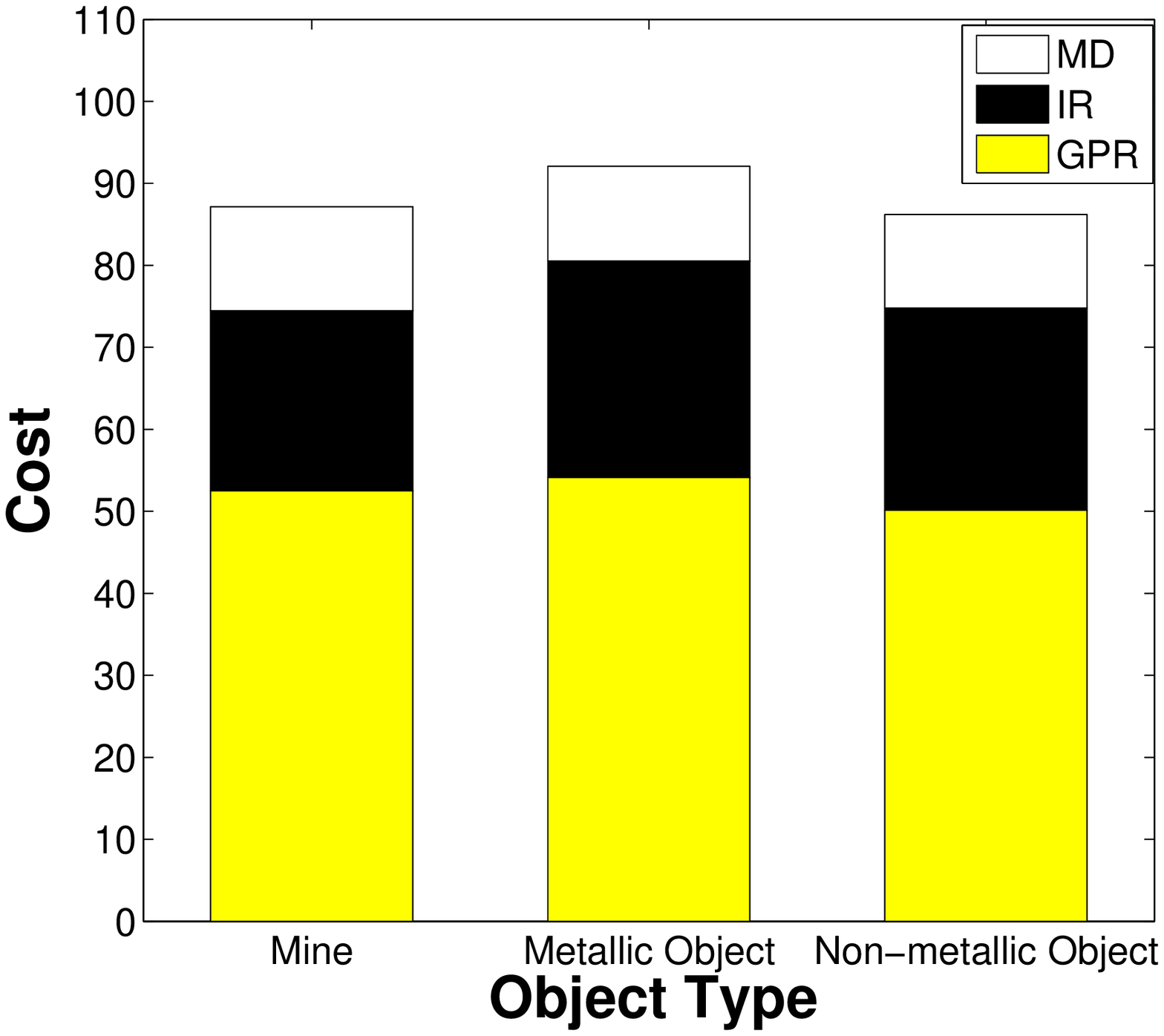}
    \hspace{-0.19in}
    &
    \hspace{-0.19in}
    \includegraphics[width=1.5in]{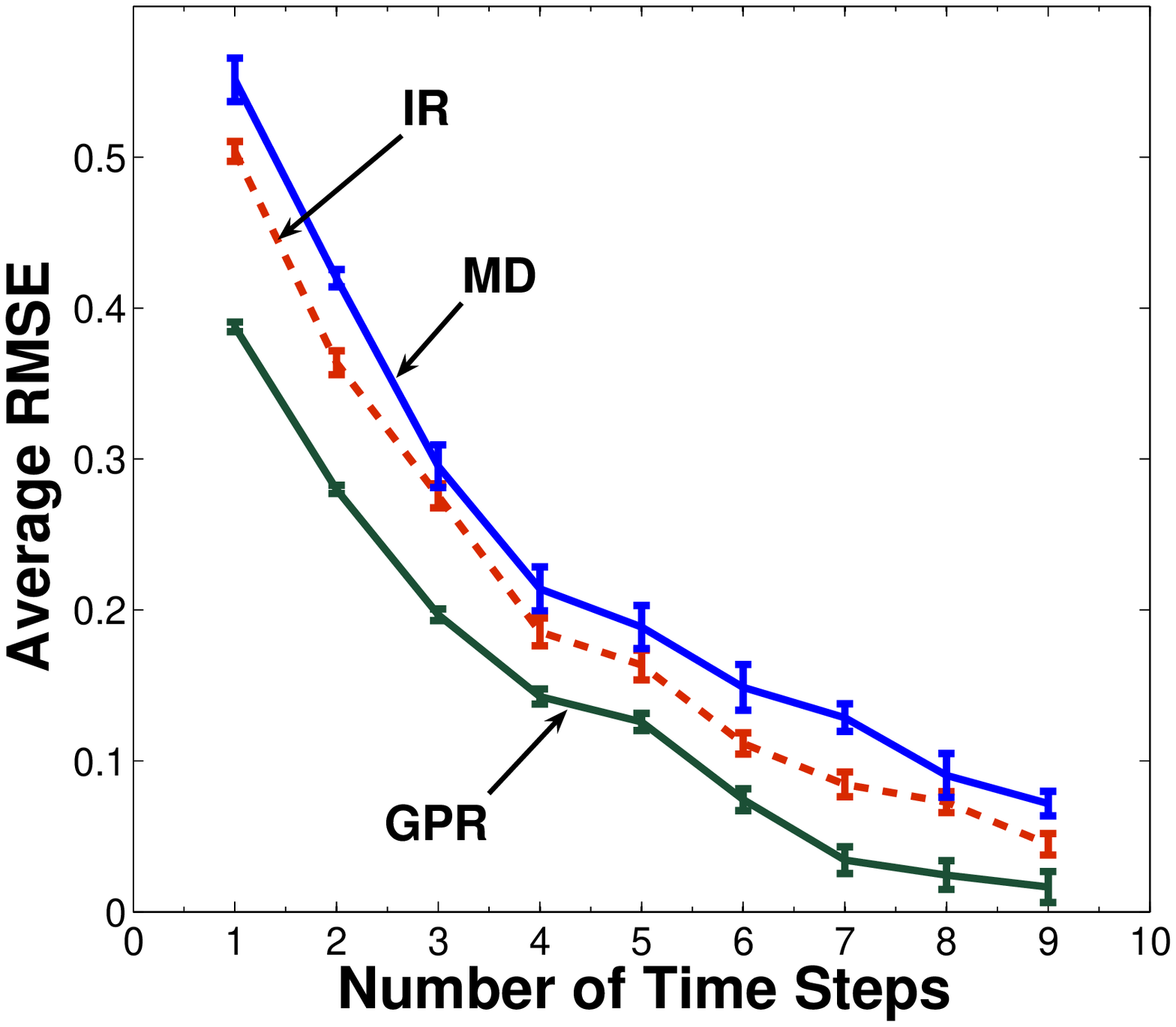}
    \hspace{-0.15in}\\
    \hspace{0.8in}a & \hspace{0.8in}b & \hspace{0.8in}c & \hspace{0.8in}d
\end{tabular}
    \caption{RMSE for different values of $w_{bel}$(a), Average sensors' utilities for different sensor types(b), Cost for different object types(c), RMSE for sensors' reports averaged over sensor types(d).}
\label{fig_sensor}
\end{center}
\end{figure*}

For our first group of experiments we analyze the performance
of our technique w.r.t. the variables in our model, such as $w_{bel}$ and time,
and, w.r.t. to sensor
and object types. We observe that as more information gets sensed
for the object, the RMSE value, shown in Figure \ref{fig_sensor}(a), decreases
over time. It takes on average $6-8$ time steps to
predict the object type with $95\%$ or greater accuracy depending on the object type
 and the value of $w_{bel}$. We also observe that our model
performs the best with $w_{bel} = 0.5$ (in Equation \ref{belief_update}),
when the agent equally
incorporates its private signal and also the market's aggregated belief
at each time step into its own belief update.
Figure \ref{fig_sensor}(b) shows the average utility of the
agents based on their type. We can see that MD sensors get
more utility because their costs of calculating and submitting reports
are generally less, whereas
GPR sensors get the least utility because they encounter the highest cost.
This result is further verified in Figure \ref{fig_sensor}(c) where we can see
the costs based on sensor types and also based on object types.
We observe that detecting a metallic
object that is not a mine has the highest cost. We posit that it is
because both MD and IR sensors can detect metallic content in the object
and extra cost is due
to the time and effort spent differentiating metallic object from a mine. Although
most of the mines are metallic \cite{Nada03,Nada01}, we can see that the cost of detecting a
mine and a non-metallic object are similar because we require a  prediction of at least $95\%$.
Due to the sensitive nature of the landmine detection problem,
it is important to ensure that even a non-metallic object
is not a mine even if we encounter higher costs.
However, despite MD's high utility (Figure \ref{fig_sensor}(b)) and
low cost (Figure \ref{fig_sensor}(c)), its error of classifying the object type is the largest, as
can be seen from Figure \ref{fig_sensor}(d).

In Table \ref{table_num_sensors}, we show how the decision
maker's decisions using our prediction market technique
results in the deployment of different numbers and types
of sensors over the time window of the object. We report
the results for the value of belief update
weight $w_{bel}=0.5$(used in Equation \ref{belief_update})
while using our prediction market model, as well
as using D-S and DDF. We see that non-metallic
object classification requires less number of sensors as
both MD and IR sensors can distinguish between
metallic vs. non-metallic objects, and so, deploying
just these two types of sensors can help to
infer that the object is not a mine. In contrast,
metallic objects require more time to get
classified as not being a mine because more object
features using all three sensor types need to be
observed. We also observe that on average our aggregation technique
using prediction market deploys a total of $6-8$ sensors and detects the object type
with at least $95\%$ accuracy
in $6-7$ time steps, while the next best compared DDF technique
deploys a total of $7-9$ sensors and detects the object type
with at least $95\%$ accuracy in $7-8$ time steps.

\begin{figure*}[ht!]
\begin{center}
\begin{tabular}{lll}
\hspace{-0.1in}
    \includegraphics[width=1.7in]{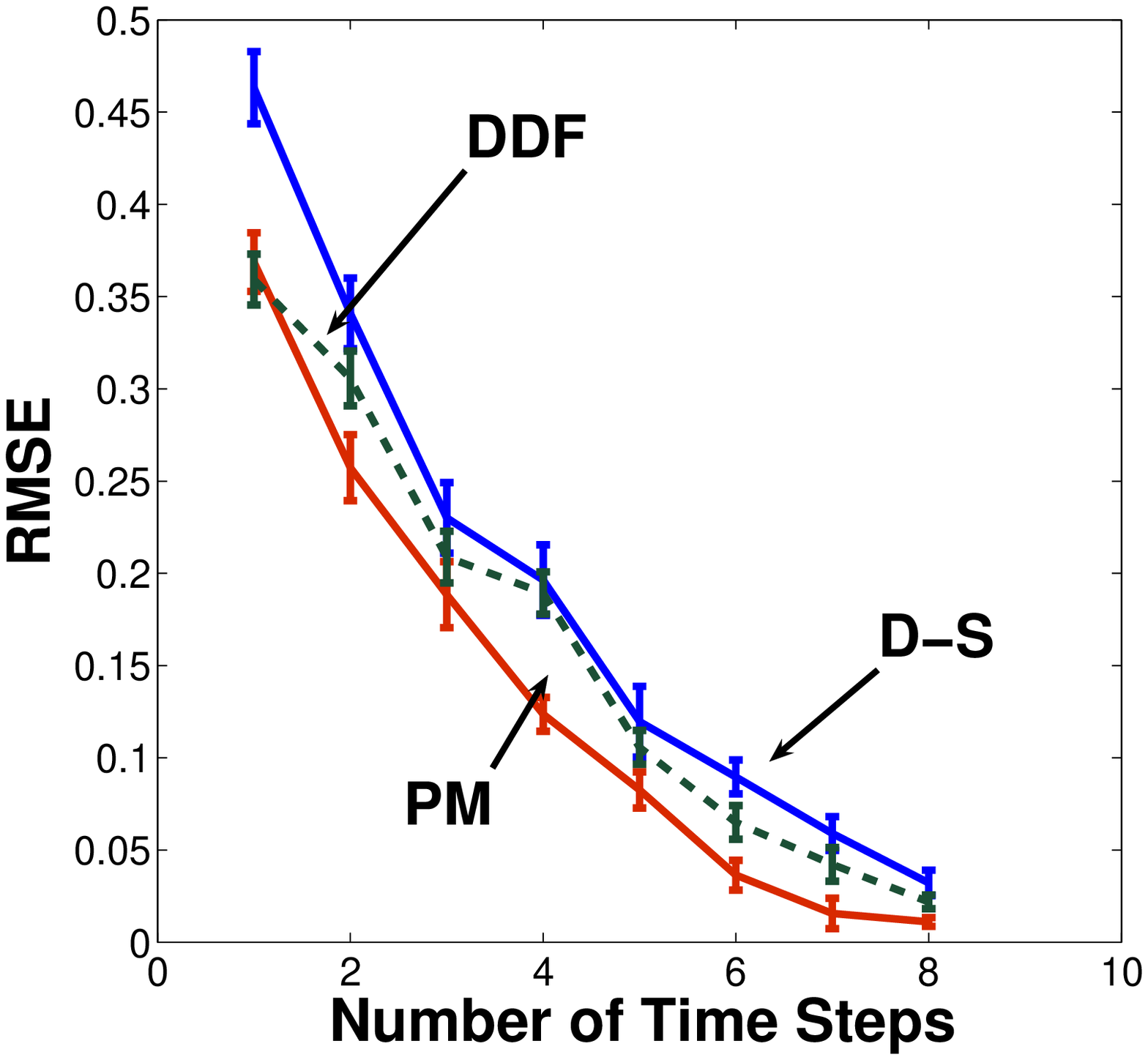}
    \hspace{-0.1in}
    &
    \hspace{-0.1in}
    \includegraphics[width=1.7in]{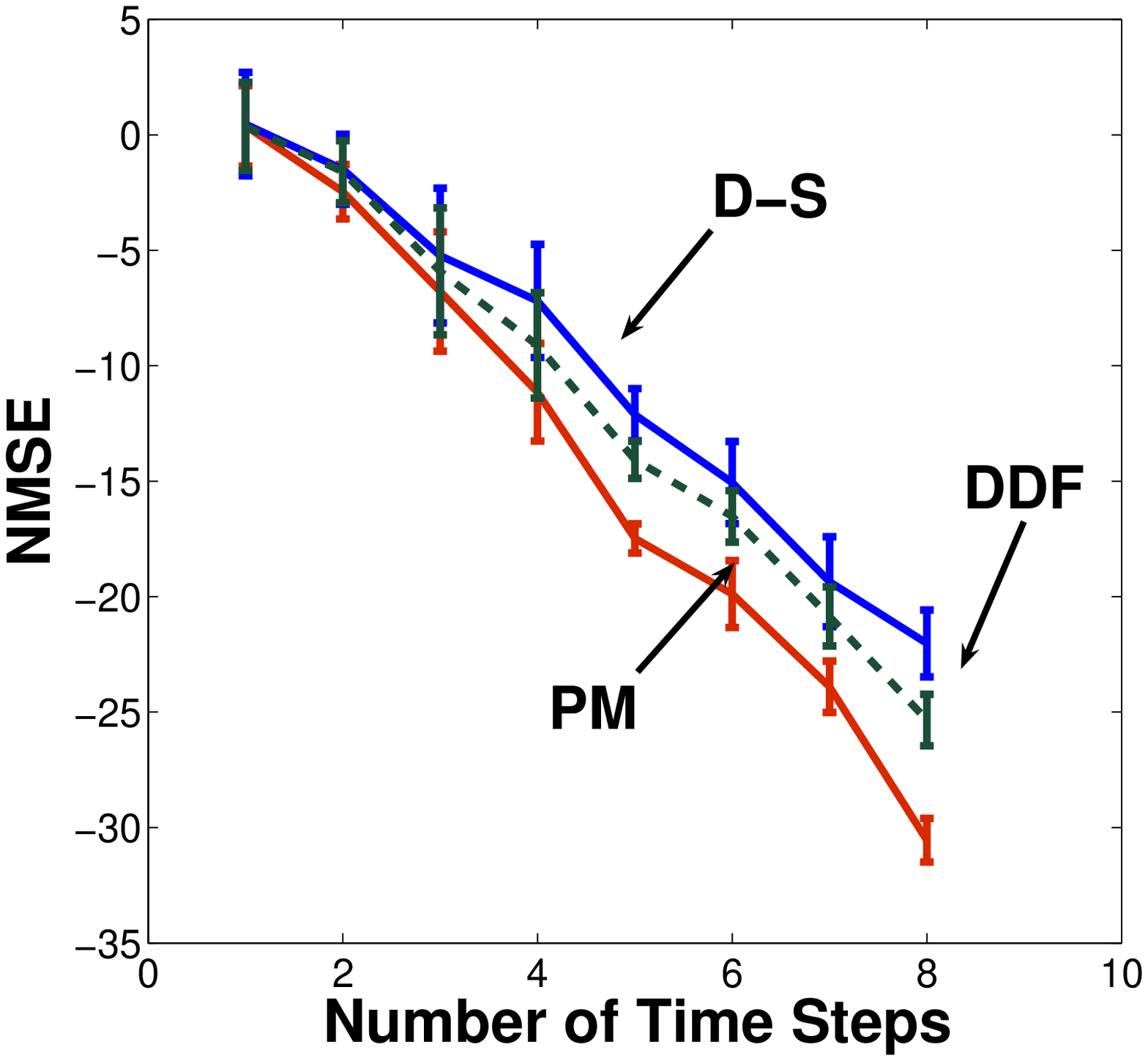}
    \hspace{-0.1in}
    &
    \hspace{-0.1in}
    \includegraphics[width=1.7in]{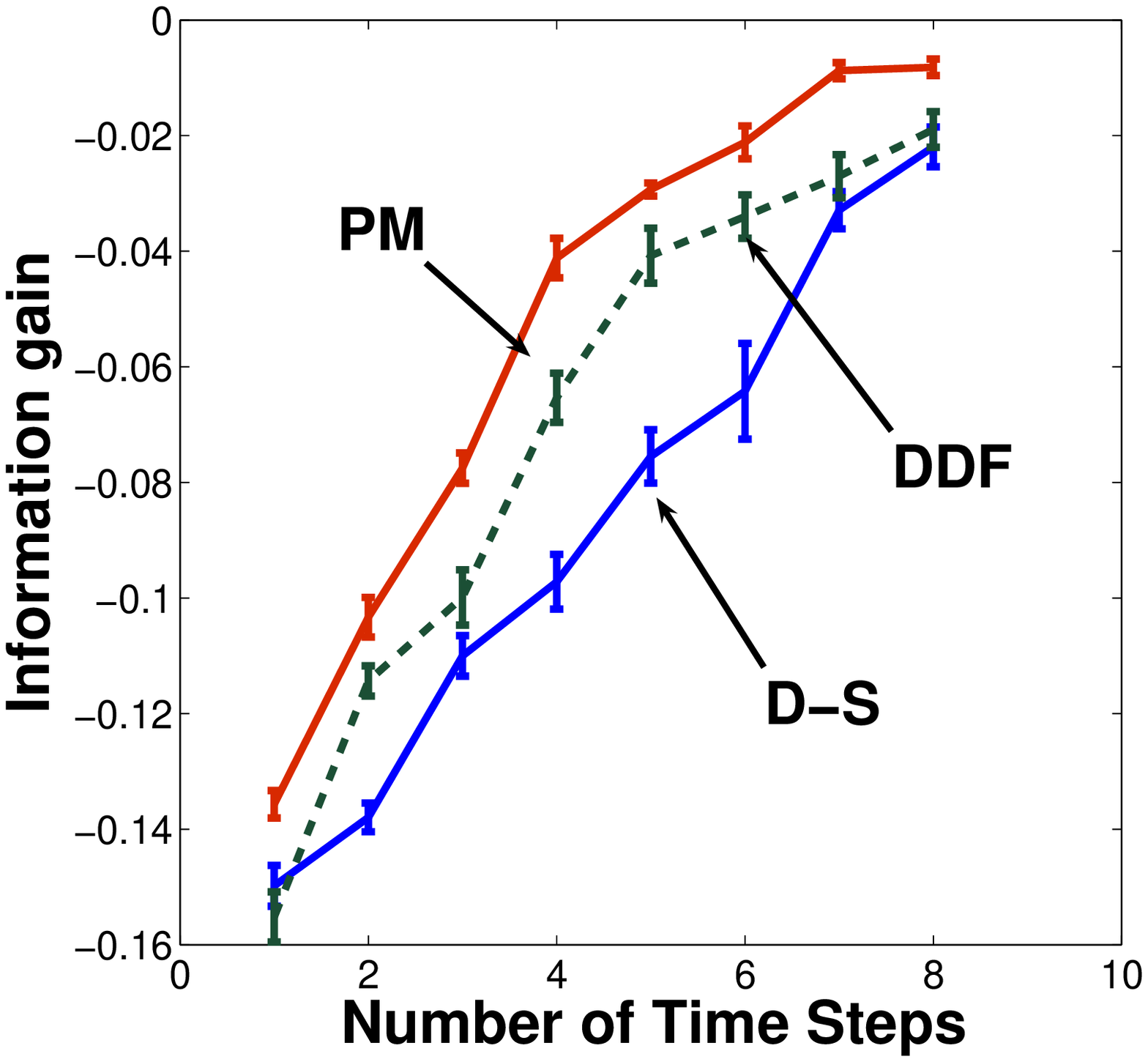}
    \hspace{-0.1in}\\
    \hspace{0.8in}a & \hspace{0.8in}b & \hspace{0.8in}c
\end{tabular}
    \caption{Comparison of our Prediction Market-based information aggregation with Dempster-Shafer and Distributed Data Fusion Techniques using different metrics: RMSE(a), NMSE(b), Information gain(c).}
\label{fig_comparison}
\end{center}
\end{figure*}

Our results shown in Figure \ref{fig_comparison}(a)
illustrate that the RMSE using our PM-based technique is
below the RMSEs using D-S and DDF by an average of
$8\%$ and $5\%$ respectively. Figure \ref{fig_sensor}(b) shows
that the NMSE values using our PM-based technique
is $18\%$ and $23\%$ less on average than D-S and DDF techniques respectively.
Finally, in Figure \ref{fig_sensor}(c) we observe that the information gain
for our PM-based technique is $12\%$ and $17\%$ more than
D-S and DDF methods respectively.

\section{Conclusions}
In this paper, we have described  a sensor information
aggregation technique for object classification
with a multi-agent prediction market and
developed a payment function used by the market maker
to incentivize truthful revelation by each agent.
Currently, the rewards given by the market maker
agent to the sensor agents are additional side payments
incurred by the decision maker. In the future we plan to investigate
a payment function that can achieve budget balance.
We are also interested in integrating our decision making problem
with the problem of scheduling robots(sensors), and,
incorporating the costs to the overall system into the decision-making costs.
Another direction we plan to investigate in the future
is a problem of minimizing the time to detect an object in
addition to the accuracy of detection.
Lastly, we plan to incorporate our aggregation technique
into the experiments with real robots.

\bibliographystyle{abbrv}
\bibliography{refs}

\end{document}